\documentclass[12pt, twoside]{article}
\usepackage{amsmath,amsthm,amssymb}{\tiny}
\usepackage{times}
\usepackage{enumerate}

\def\titlerunning#1{\gdef\titrun{#1}}
\makeatletter
\def\author#1{\gdef\autrun{\def\and{\unskip, }#1}\gdef\@author{#1}}
\def\address#1{{\def\and{\\\hspace*{18pt}}\renewcommand{\thefootnote}{}%
\footnote {#1}}%
\markboth{\autrun}{\titrun}}
\makeatother
\def\email#1{e-mail: #1}
\def\subjclass#1{{\renewcommand{\thefootnote}{}%
\footnote{\emph{Mathematics Subject Classification (2010):} #1}}}
\def\keywords#1{\par\medskip
\noindent\textbf{Keywords.} #1}

\theoremstyle{definition}

\numberwithin{equation}{section}

\def\tensor{\,\raise2pt\hbox{${}_{\otimes}$}\,}
\def\fdg{\,:\,}
\def\ptl{\partial}
\def\rest#1{\raise-2pt\hbox{${\lfloor_{#1}}$}}

\def\mbo#1{\boldsymbol{#1}}

\def\ip#1#2{\langle#1,#2\rangle}

\def\olin#1{\overline{#1}{}}

\def\grad{{\nabla}}

\def\halb{\frac{1}{2}}

\def \a{\alpha}
\def \b {\beta}

\newtheorem{theorem}{Theorem}[section]
\newtheorem{lemma}[theorem]{Lemma}

\newcommand{\ba}{\begin{array}}
\newcommand{\ea}{\end{array}}
\newcommand{\bea}{\begin{eqnarray}}
\newcommand{\eea}{\end{eqnarray}}
\newcommand{\bee}{\begin{eqnarray*}}
\newcommand{\eee}{\end{eqnarray*}}

\renewcommand{\vec}[1]{\mbox{\boldmath $#1$}}

\renewcommand{\a}{\alpha}
\renewcommand{\b}{\beta}

\usepackage{color}

\newcommand{\green}[1]{{\color{green}#1}}

\newcounter{mnotecount}[section]

\renewcommand{\themnotecount}{\thesection.\arabic{mnotecount}}

\newcounter{mymnotecount}[section]
\renewcommand{\themymnotecount}{\thesection.\arabic{mymnotecount}}
\newcommand{\mymnote}[1]{\protect{\stepcounter{mymnotecount}}${\raisebox{0.5\baselineskip}[0pt]{\makebox[0pt][c]{\color{green}{\tiny\em$\bullet$\themnotecount}}}}$\marginpar{\raggedright\tiny\em$\!\bullet$\themymnotecount:

\green{#1}}\ignorespaces}

\renewcommand{\mymnote}[1]{}

\begin{document}


\baselineskip=17pt


\titlerunning{Quasi Local Mass}

\title{Quasi-Local Mass near the Singularity,  the Event Horizon and the Null Infinity of Black Hole Spacetimes}

\author{Nishanth Gudapati
\and 
Shing-Tung Yau}

\date{}

\maketitle

\address{N.Gudapati: Center of Mathematical Sciences and Applications, Harvard University, 20 Garden Street, Cambridge, MA-02138, USA; \email{nishanth.gudapati@cmsa.fas.harvard.edu}
\and
S.-T. Yau: Department of Mathematics, Harvard University, 1 Oxford Street, MA-02138, USA; \email{yau@math.harvard.edu}}



\begin{abstract}
The behaviour of geometric quantities close to geometric pathologies of a spacetime is relevant to deduce the physical behaviour of the system. In this work, we compute the quasi-local mass quantities -  the Hawking mass, the Brown-York mass and the Liu-Yau mass  in the maximal extensions of the spherically symmetric solutions of the Einstein equations inside the black hole region, at the singularity,  the event horizon, and the null infinity, in the limiting sense of a geometric flow. 
\end{abstract}
\section{Quasi-Local Mass and Black Holes}
The notion of gravitational mass-energy  plays an important role in studying the physical properties of a spacetime $(\bar{M}, \bar{g})$. For instance, a linear wave equation $(\square_{\bar{g}} u =0)$ provides a useful tool to study the physically relevant aspects of spacetimes such as stability, red-shift behaviour near black hole spacetimes (e.g., Schwarzschild, Kerr black hole spacetimes).
In contrast with a linear wave equation, which has a well-defined notion of energy, e.g.,
\begin{align}
E \fdg = \int_{\olin{\Sigma}} ( \vert \ptl_t u \vert^2 + \vert \grad u \vert_x^2 ) \bar{\mu}_{\bar{q}}
\end{align}
 and energy density $ \mathbf{e} \fdg = \ptl_t u^2 + \vert \grad u \vert_x^2$ on a spacelike hypersurface $(\olin{\Sigma}, \bar{q}),\bar{M}=\olin{\Sigma}\times~\mathbb{R}.$  In general relativity, an unambiguous notion of local mass-density is not reasonable due to the equivalence principle. Therefore, a quasi-local notion of mass is used for Einstein's equations for general relativity, where the mass-energy is captured in  the inner boundary $\Sigma$ of a spacelike hypersurface $\olin{\Sigma} \hookrightarrow \bar{M},$ where $\bar{M}$ is a $3+1$ dimensional Lorentzian spacetime. Suppose, the metric $\bar{q}$ on $\olin{\Sigma}$ is given by 
 \begin{align}
 \bar{g} = - \bar{N}^2 dt^2 + \bar{q}_{ab} (dx^i + \bar{N}^i dt) \otimes  (dx^j + \bar{N}^j dt ), 
 \end{align}
 where $\bar{N} \neq 0,$ so that $(\olin{\Sigma}, \bar{q})$ is a (spacelike) Riemannian hypersurface. 
 \noindent Let us first introduce the notion of ADM mass at spatial infinity. In an asymptotically flat spacetime, where $\olin{\Sigma}$ is such that, outside a compact set of $\olin{\Sigma}$ it is diffeomorphic to $\mathbb{R}^3 \setminus B_1(0)$  and has the following asymptotic behaviour 
 \begin{subequations}
 	\begin{align}
 	\bar{q}_{ij} =& \left(1+\frac{M}{r} \right) \delta_{ij} + \mathcal{O}(r^{-1-\a}) \label{q-asym}
 	\intertext{and}
 	\bar{K}_{ij} =& \mathcal{O}(r^{-2-\a}). 
 	\end{align}
 \end{subequations}
 It may be noted that the parameter $M$ in \eqref{q-asym} is also the ADM mass defined as 
 \begin{align}
 M \fdg = \lim_{r \to \infty} \int_{\mathbb{S}^2(r)} (\ptl_k \bar{q}_{i \ell} - \ptl_i \bar{q}_{\ell k}) \frac{\vert x \vert^i}{r}  \bar{\mu}_{\bar{q}}.
 \end{align}
 This (total) mass is computed at the outer boundary near the asymptotically flat end. 
 We have $M \geq 0,$ with the equality  $M=0$ iff $\olin{\Sigma} = \mathbb{R}^3$ (Euclidean space), from the famous works of Schoen-Yau \cite{schoen-yau-1,schoen-yau-2} and Witten \cite{witten-pmt}.
 
 In this work, we shall be interested in various notions of quasi-local mass, which are defined as functionals on the (boundary) 2-surface $\Sigma$ such that 
 $\Sigma \hookrightarrow \olin{\Sigma}.$ Now suppose that the embedding $\Sigma \hookrightarrow \olin{\Sigma},$ is such that $A$ is the second fundamental form and $H$ is the mean curvature ($H \fdg = \text{tr}(A)$). Let us start with the notion of Hawking mass, 
 
\begin{align} \label{haw-orig}
 m_{\text{H}} (\Sigma) \fdg = \frac{\vert \Sigma \vert^{\halb}}{(16\pi)^{3/2}} \left( 16 \pi - \int_{\Sigma} H^2 \right)
\end{align}

\noindent  A further notion of quasi-local mass is given by the Brown-York mass \cite{BY_93}
 
 \begin{align}
 m_{\text{BY}} \fdg = \frac{1}{8 \pi}\int_{\Sigma} (H_0 - H), \quad \text{(Brown-York mass)}
 \end{align}
 where $H_0$ is the mean curvature of the isometric embedding of $\Sigma$ in the Euclidean space. The notion of the Brown-York mass was extended by Liu-Yau \cite{liu-yau-1, liu-yau-2},

 \begin{align}
 m_{\text{LY}} \fdg = \frac{1}{8 \pi}\int_{\Sigma} ( H_0  - \vert  \mathbf{H} \vert), \quad \text{(Liu-Yau mass)} 
 \end{align}
 where $\mathbf{H}$ is the mean curvature vector, which follows from the first variation of the area $\delta (\vert \Sigma \vert)$ and $H_0$ is the mean curvature  of the isometric embedding of $\Sigma$ into  the Minkowski space. The definition of the Hawking mass \eqref{haw-orig} can also be generalized as follows: 
\begin{align} \label{spacetime-Hawking}
m_{\text{H}} \fdg=  \frac{\vert \Sigma \vert^{\halb}}{(16\pi)^{3/2}} \left( 16 \pi - \int_{\Sigma} \vert \mathbf{H} \vert^2 \right) \quad \text{(spacetime Hawking mass).}
\end{align}
 
  A further notion of quasi-local mass is defined by Wang-Yau \cite{Wang2009}, which is based on the optimal isometric embeddings of $\Sigma$ into the Minkowski space. 
 Now consider a general flow of the hypersurfaces $(\Sigma_{\tau}, q) \hookrightarrow (\olin{\Sigma}, \bar{q}).$ Suppose we start with a general flow as follows 
\begin{align} \label{q-gen-flow}
\dot{q}_{ab} = 2 u A_{ab}
\end{align} 
where a dot signifies a derivative along the flow paramater $\tau,$ whose level sets are the hypersurfaces $\Sigma_{\tau}$ and $A$ is the second fundamental form of the embedding $\Sigma_{\tau} \hookrightarrow \olin{\Sigma}$.
\begin{theorem}
Consider a geometric flow of $2-$surfaces $\Sigma_t \hookrightarrow \olin{\Sigma}$ and satisfies \eqref{q-gen-flow}, then 
\begin{enumerate}
\item The evolution of the mean curvature $H$ is given by
\begin{align}
\dot{H} = - \grad^a \grad_a u +  \halb u \,\, (-H^2 - \Vert A \Vert^2_q  - R_{\bar{q}} + R_q )
\end{align}
where $H$ and $A$ are the mean curvature and extrinsic curvature of $\Sigma_ \tau \hookrightarrow \olin{\Sigma}$ respectively, $R_q$ is the scalar curvature of $\Sigma$ and $R_{\bar{q}}$ is the scalar curvature of $\olin{\Sigma}.$
\item The evolution of the Hawking mass $m_{H}$ is given by
\begin{align}
\dot{m}_{\text{H}} =& \frac{1}{2(16 \pi)^{3/2}} u \vert \Sigma \vert^{1/2} H \left(16 \pi - \int_{\Sigma} H^2 \bar{\mu}_q \right)  \notag\\
&+ \frac{\vert \Sigma \vert^{1/2}}{(16 \pi)^{3/2}} \int_{\Sigma} \Big( 2H \bar{\mu}_q \grad^a \grad_a u- \halb uH \tilde{w} + uH \bar{\mu}_q (A^{ab} - \halb H q^{ab}) (A_{ab} - \halb Hq_{ab}) \notag\\
&+ u H \bar{\mu}_q (R_{\bar{q}} -R_q) \Big) 
\end{align}
\item There exists a choice of the function $u$ of the form $u = \frac{1}{H},$  
such that $m_{\text{H}}$ is monotonic with respect to the flow \eqref{q-gen-flow}
\end{enumerate}  
\end{theorem}
\begin{proof}
As we already remarked the mean curvature $H$ closely related to the first variation of area $\delta(\vert \Sigma \vert).$ It follows that
\begin{align}\label{H-gen-flow}
\dot{H} = - \grad^a \grad_a u -  u \Vert A \Vert_q^2  -   u \,\, \text{Ric} (n, n)
\end{align}
where $\text{Ric}$ is the Ricci curvature of $\olin{\Sigma}$ and $n$ is the unit normal of $\Sigma$ in $\olin{\Sigma}.$ This computation, in consistency with \eqref{q-gen-flow}, is closely related to the second variation formula of Schoen-Yau \cite{schoen-yau-1}. Then, using the Gauss-Kodazzi relations between curvatures of $(\olin{\Sigma}, \bar{q})$ and $(\Sigma, q),$ we have 
\begin{align}\label{GK}
\text{Ric}(n, n) = \halb (R_{\bar{q}} -R_q - \Vert A \Vert^2_q+ H^2 ).
\end{align}
Now then plugging in \eqref{GK} in \eqref{H-gen-flow}, we get the flow equation:
\begin{align}
\dot{H} = - \grad^a \grad_a u +  \halb u \,\, (-H^2 - \Vert A \Vert^2_q - R_{\bar{q}} + R_q ). 
\end{align}

\end{proof}

Consider the quantity: 
\begin{align}
w \fdg= - H^2 \bar{\mu}_q 
\end{align}
Now then, if we compute the evolution equation of $w$ for the general flow \eqref{q-gen-flow}: 
we have

\begin{align}
 \dot{w} =  2 \bar{\mu}_q H(  \grad^a \grad_a u + \halb u ( \Vert A \Vert^2_q + H^2 + R_{\bar{q}} - R_{q}  )) - H^2 (u \bar{\mu}_q H).  
\end{align}
Using the divergence identity, 
\begin{align}
\grad^a (H \grad_a u) = \grad^a H \grad_a u + H \grad^a \grad_a u
\intertext{then, under an integral over a compact manifold $\Sigma$}
 H \grad^a \grad_a u= - \grad^a H \grad_a u. 
\end{align} 
we get 
\begin{align}
\dot{w} = - \bar{\mu}_q \grad^a H \grad_a u + u \bar{\mu}_q H \Vert A \Vert^2_q + u \bar{\mu}_q H (R_{\bar{q}} - R_{q})
\end{align}
We have the following decomposition: 
\begin{align}
\Vert A \Vert_q^2 - \halb H^2 =& A^{ab} A_{ab} - \halb H^2 - \halb H^2 + \halb H^2 \notag\\
=& (A^{ab} - \halb H q^{ab}) (A_{ab} - \halb H q_{ab})
\end{align}
which is a positive-definite quantity. Now let us find an expression for the rate of change of $w$ along the flow.   We have
\begin{align}
\dot{w}=& 2 \bar{\mu}_q H \grad^a \grad_a u- \halb uH w + uH \bar{\mu}_q (A^{ab} - \halb H q^{ab}) (A_{ab} - \halb Hq_{ab}) \notag\\
&+ u H \bar{\mu}_q (R_{\bar{q}} -R_q)
\intertext{equivalently}
\dot{w}=& -2 \bar{\mu}_q \grad^a H \grad_a u - \halb uH w + uH \bar{\mu}_q (A^{ab} - \halb H q^{ab}) (A_{ab} - \halb Hq_{ab}) \notag\\
&+ u H \bar{\mu}_q (R_{\bar{q}} -R_q)
\end{align}
If we define, $\tilde{w} = \bar{\mu}_q (2R_{q} -H^2),$
\begin{align}
\dot{w} =& 2 \bar{\mu}_qH \grad^a \grad_a u- \halb uH \tilde{w} + uH \bar{\mu}_q (A^{ab} - \halb H q^{ab}) (A_{ab} - \halb Hq_{ab}) + uH \bar{\mu}_q R_{\bar{q}}
\intertext{or}
\dot{w} =& -2 \bar{\mu}_q \grad^a H \grad_a u- \halb uH \tilde{w} + uH \bar{\mu}_q (A^{ab} - \halb H q^{ab}) (A_{ab} - \halb Hq_{ab}) + uH \bar{\mu}_q R_{\bar{q}}.
\end{align}
It may be noted that the quantity $u$ determines the velocity of the geometric flow, a particular choice of which may provide the desired properties that we need. 
\noindent It is evident that if we consider $u = \frac{1}{H}$ we get the inverse mean curvature flow and
it can be established that the Hawking mass $m_{\text{H}}$

 \begin{align}
m_{\text{H}} \fdg = \frac{\vert \Sigma \vert^{1/2}}{(16 \pi)^{3/2}} \left(16 \pi - \int_{\Sigma} H^2 \bar{\mu}_q \right)
\end{align}
is monotonic under this flow \cite{Geroch_73, huisken2001}. In general, the evolution of $m_{\text{H}}$ as per the flow \eqref{q-gen-flow} is as follows 

\begin{align}
\dot{m}_{\text{H}} =& \frac{1}{2(16 \pi)^{3/2}} u \vert \Sigma \vert^{1/2} H \left(16 \pi - \int_{\Sigma} H^2 \bar{\mu}_q \right)  \notag\\
&+ \frac{\vert \Sigma \vert^{1/2}}{(16 \pi)^{3/2}} \int_{\Sigma} \Big( 2H \bar{\mu}_q \grad^a \grad_a u- \halb uH \tilde{w} + uH \bar{\mu}_q (A^{ab} - \halb H q^{ab}) (A_{ab} - \halb Hq_{ab}) \notag\\
&+ u H \bar{\mu}_q (R_{\bar{q}} -R_q) \Big) \\
\dot{m}_{\text{H}} =& \frac{1}{2(16 \pi)^{3/2}} u \vert \Sigma \vert^{1/2} H \left(16 \pi - \int_{\Sigma} H^2 \bar{\mu}_q \right)  \notag\\
&+ \frac{\vert \Sigma \vert^{1/2}}{(16 \pi)^{3/2}} \int_{\Sigma} \Big( -\grad^a H \grad_a u- \halb uH \tilde{w} + uH \bar{\mu}_q (A^{ab} - \halb H q^{ab}) (A_{ab} - \halb Hq_{ab}) \notag\\
&+ u H \bar{\mu}_q (R_{\bar{q}} -R_q) \Big).
\end{align}
It follows that $m_{\text{H}}$ is nondecreasing for the choice of $u = \frac{1}{H}$, $\Sigma$  connected and non-negative scalar curvature $R_{\bar{q}} \geq 0$ of the ambient manifold $(\olin{\Sigma}, q).$  However, from a geometric PDE perspective, establishing the regularity and monotonicity of the inverse mean curvature flow is quite delicate, especially in the presence of minimal surfaces. The regularity theory of inverse mean curvature flow and Riemannian Penrose inequality was established by Huisken-Ilmanen \cite{huisken2001}. A separate proof of the Riemannian-Penrose inequality was established by Bray \cite{bray2001}, using conformal flow of metrics. 
 
 Now let us consider the Schwarzschild  spacetimes. In our work, we are interested in the limits of the geometric quantities such as the Hawking mass, the Brown-York mass and the Liu-Yau mass close to geometric pathologies such as the event horizon and the singularity. Since these quantities are defined over a 2-surface $\Sigma$ we need to interpret these limits in the sense of a geometric flow.  
 
 Let us first start with the representation in the domain of outer communications.  Later, we shall consider the (Kruskal) maximally extended versions of these spacetimes. 
 \begin{align}\label{Sch-doc}
 \bar{g} = -f dt^2 + f^{-1} dr^2 + r^2 d \omega_{\mathbb{S}^2}, \quad r>2M
 \end{align}
  where $f = (1-\frac{2M}{r}).$  In the gauge  \eqref{Sch-doc} used for the Schwarzschild metric, the black hole region is given by $r <2M$, $r=2M$ is the event horizon and $r>2M$ is the domain of outer communications. The intersection of a spacelike hypersurface $\olin{\Sigma} = \{ t= const. \}$ with $r=2M$ is a minimal surface. The Schwarzschild spacetime is also a conformally flat spacetime: 
  \begin{align} \label{sch-iso}
  \bar{q}_{ij} = \left( 1+ \frac{M}{2r^*} \right)^4 \bar{\delta}_{ij}
  \end{align}
  We would like to point out that there is a \emph{coordinate} singularity in the gauge \eqref{Sch-doc} at $r=2M$ but $r=0$ is an actual \emph{geometric} singularity. The function $r^*$ in the isotropic coordinates is related to the form of the Schwarzschild  as
\begin{align}
r = \left(1+\frac{m}{2r^*}\right)^2 r^*, \quad r^* = \frac{M}{2} \quad \text{is the horizon}
\end{align}
Let us compute the mean curvature of the embedding formed by $r = const.$ hypersurfaces $(\Sigma) \hookrightarrow \olin{\Sigma}.$ Following Schoen-Yau \cite{schoen-yau-1}, we have the following relation between the mean curvature $H$ and $H_0$ under the conformal transformation $q = \Omega^4 q_0$
\begin{align}
H = \frac{H_0}{\Omega^2} + 4 \frac{\nu_0 (\Omega)}{\Omega^3}
\end{align}
which, for the metric \eqref{sch-iso}, turns out to be 
\begin{align} \label{H-sch-iso}
H = 	\frac{(2r^* -M)}{r^*{^2}(1+\frac{M}{2r^*})^3}, \quad \text{(isotropic coordinates)}
\end{align}
for $H_0 = \frac{2}{r^*},$ in the Euclidean space. The function $H$ is not a monotonic function with respect to $r^*$ and it has a maximum at the photon sphere of the Schwarzschild spacetime. In these coordinates, 
the Brown-York mass can be computed as
 \begin{align}
 m_{\text{BY}} =& \frac{1}{8 \pi} \int_{\Sigma} H_0 - H \notag\\
  =&M \left(1+ \frac{M}{2 r^*} \right) 
 \end{align}
in view of \eqref{H-sch-iso}. Likewise, in the maximal gauge, we have 
\begin{align}
H= \frac{2}{r} \sqrt{1-\frac{2M}{r}}, \quad r \geq 2M
\end{align} 
As a consequence, the Hawking mass 
\begin{align}
m_{\text{H}} = \frac{r}{2} (1-f)
\end{align}
and the Brown-York mass, 
\begin{align}
m_{\text{BY}} = r - r \sqrt{1-\frac{2M}{r}}, \quad r \geq 2M
\end{align}
We would like to remark that, for the isometric embedding of $\Sigma$ in the Euclidean space, $H_0 = \frac{2}{r},$ therefore, we have $\int_{\Sigma} H_0 = r.$ In  the time-symmetric gauge \eqref{Sch-doc} of the Schwarzschild spcaetime, the Liu-Yau mass is the same as the Brown-York mass: 
\begin{align}
m_{\text{LY}} = m_{\text{BY}}.
\end{align}
The time-symmetric gauge used in \eqref{Sch-doc} holds for the exterior region of black hole spacetimes.  
  To study the quasi-local mass quantities in the interior, we work in a geometric framework where 
  the geometry can be smoothly extended into the interior. In the case of Schwarzschild black holes, the maximal extension contains the region considered in \eqref{Sch-doc} as a proper subset. 
 
\subsection{The Maximal Extension and the Interior of the Schwarzschild Solution}
In order to prevent a coordinate singularity at the event horizon $\{ r=2M \},$  the Kruskal coordinates were introduced: 
\begin{align}
\bar{g} = \mathring{g}_{\mu \nu} dx^\mu \otimes dx^\nu + r^2 d\omega_{\mathbb{S}^2},\quad  \text{on the maximally extended Schwarzschild spacetime}
\end{align}
where
\begin{align}
\mathring{g} =& - e^{2Z} d U dV, \quad \text{is the metric on the quotient} \quad \bar{M}/SO(3), 
\intertext{which can be rewritten in $(T, R)$ coordinates as}
\mathring{g} =& e^{2Z} (-dT^2 + dX^2),
  \intertext{where (T, R) coordinates are such that} 
T^2-X^2 =& \left(1-\frac{r}{2M} \right)e^{r/2M}, \frac{T+X}{-T+X}= e^{t/2M}, \quad e^{2Z} = \frac{32M^3 e^{-\frac{r}{2M}}}{r}. 
\end{align}
The advantage of the `Kruskal extension' is that it covers the domain of outer communications $r > 2M$, the event horizon $r=2M$  as well as  the interior region $0< r<2M$ of maximally extended Schwarzschild spacetimes.  

In this work we shall be interested in the exterior and the interior region of the maximally extented Schwarzschild black hole spacetime. Constant mean curvature hypersurfaces   that admit smooth extensions into the maximal Schwarzschild are studied in \cite{BCIsen_80, MM_03, MM_09,Lee-Lee_16}.
This construction is useful to us because, this shall allow us to compute the quasi-local mass quantities for 2-surfaces in the maximally extended spacetimes. 

 Let us start with the Schwarzschild spacetime in the usual time-symmetric gauge: 

\begin{align} \label{sch-orig}
\bar{g} = - f dt^2 + f^{-1} dr^2 + r^2 d\omega^2_{\mathbb{S}^2}.
\end{align}
In this work, we shall follow the mathematical construction of Lee-Lee \cite{Lee-Lee_16}. We also refer the reader to the works of Brill-Cavallo-Isenberg\cite{BCIsen_80} and Malec-Murchadha \cite{MM_03, MM_09} 
 for the background and previous works on this subject. 
A height function $h$ is introduced such that $\hat{t} =-t+h$ and $ \{ \hat{t} = const.\}$ form  `spacelike' constant mean curvature hypersurfaces.
As a consequence, a priori we have the condition that

\begin{align}
\bar{g}^{\mu \nu} \ptl_\mu \hat{t} \ptl_\nu \hat{t} < 0, \quad \text{(i.e., the unit normal $n = \frac{1}{\Vert \grad \hat{t}\Vert_{\bar{g}}}\grad \hat{t}$ is timelike)}.
\end{align} 
 As we shall see later, the advantage of a radial `height function' $h(r)$ is that the geometry of $\Sigma$ is not disturbed and this in turn simplifies the isometric embeddings into the Euclidean and Minkowski space used in the construction of quasi-local mass quantities (e.g., Brown-York, Liu-Yau, Wang-Yau etc).  The sphere of symmetry metric $(\Sigma)$ can be expressed in the orthonormal frame: 
\begin{align}\label{ortho-orig}
e_1=& \frac{1}{r} \ptl_\theta, \quad e_2 = \frac{1}{r \sin \theta} \ptl_\phi
\intertext{so that}
q =& e_1 \otimes e_1 + e_2 \otimes e_2, \quad  \text{on} \quad (\Sigma)
\end{align}

In the usual Schwarzschild time-symmetric gauge, we have for the $\olin{\Sigma} \hookrightarrow M$ embedding the second fundamental form $K=0$ and the mean curvature is also $0$. Let us now compute these quantities for the $\olin{\Sigma}_{\hat{t}} \hookrightarrow M,$ for the height function $h.$ 
Now consider the quantity $f^{-1} -h'^2 f,$ we have
\begin{align}
(f^{-1}  - h'^2  f) 
=\left( \frac{r}{r-2M} + h'^2\frac{(r-2M)}{r} \right)  >0 \quad \text{in the region} \quad r>2M.
\end{align}
The remaining components of the Schwarzschild metric in the orthonormal frame are given by: 
\begin{align}
e_3 = \frac{h'}{ (f^{-1} - h'^2 f)^{1/2}} \ptl_r, \quad e_4 = \frac{1}{ (f^{-1} - h'^2 f)^{1/2}} \left( f^{-1}\ptl_t + f h' \ptl_r \right) 
\end{align}
 Firstly, we have the spacetime second fundamental form
\begin{align}
\hat{K}_{ij} = \ip{\grad_{e_i} n}{e_j}, \quad i, j = 1, 2, 3. 
\end{align}

Now then,  with the notation $h' =\ptl_r h $ and $h''=\ptl^2_r h$
\begin{align}
\hat{K}_{11} = \hat{K}_{22} =& \frac{ (r-2M) h'}{r^2 (f^{-1} - h'^2 f)^{\halb}}, \\
 \hat{K}_{33} =& \frac{1}{(f^{-1} -h'^2 f)^{3/2}} \left(  h'' + \frac{3 h'f'}{2f} - \frac{h'^3 ff'}{2} \right)
\end{align} 
which is  `diagonal' in our orthonormal frame. In particular, the 2-surface $\Sigma$ is umbilical. We would like to point out that, in general, the conditions on the function associated to a warped product space, for which a constant mean curvature 2-surface $\Sigma$ is umbilical, was established in \cite{Brendle_13} by Brendle. This result is a generalization of the classic Alexandrov theorem.  
Now if we compute the mean curvature $\hat{H},$ we have 
\begin{align}
\hat{H} =& \frac{1}{3 (f^{-1} - h'^2 f)^{1/2}} \left(  \frac{h''}{f^{-1} - h'^2 f} + (\frac{2f}{r}  + \frac{f'}{2})h' + \frac{f'}{f (f^{-1} - h'^2 f)} h' \right) \\
\intertext{explicitly}
 =&\frac{1}{3 (f^{-1} - h'^2 f) ^{\halb}} \left( \frac{h''}{f^{-1} - h'^2 f} + \frac{2 (r-M)}{r} h' + \frac{2M}{r (r-2M) (f^{-1} -h'^2f)} h'\right)
\end{align}
This can be converted to a differential equation for $h$ (cf. \cite{Lee-Lee_11, Lee-Lee_16})
\begin{align} \label{ode-f}
h'' + (f^{-1} - h'^2 f) \left(\frac{2f}{r} + \frac{f'}{2} \right) h' + \frac{f'}{f} h - 3 \hat{H} (f^{-1} - h'^2 f)^{3/2} =0
\intertext{explicitly}
h'' + \frac{2M}{r (r-2M)}h' + \frac{2 (f^{-1} - h'^2 f)(r - M)}{r^2} h' -3 \hat{H}  (f^{-1} - h'^2 f)^{3/2}=0.
\end{align} 
In the context of the constant mean curvature hypersurfaces,  the differential equation  \eqref{ode-f} for the unknown function $h$ is posed in such a way that $\hat{H}=const.$ and we are interested in such solutions $h$ that solve \eqref{ode-f}. Let us first focus on the domain of outer communications $(r>2M)$.  This equation  $h$ can be reduced to a separable form  using the transformation $\sin \eta(r) = h_1' \frac{r-2M}{r}$ for $(-\frac{\pi}{2}, \frac{\pi}{2}).$
Now then, noting that
\begin{align}
h'' = \frac{r}{r-2M} (\cos\eta \, \eta' + \frac{2M}{r^2} h'), \quad \eta \in (-\frac{\pi}{2}, \frac{\pi}{2}),
\end{align}
we get
\begin{align}
(\tan \eta )' + \frac{1}{r}\left(2 + \frac{2 M}{(r-2M)} \right) \tan \eta - 3 \hat{H} \left(  \frac{r}{r-2M}\right)^{\halb} =0
\end{align}
Now then, 
\begin{align}  \label{tan-eta-soln}
\tan \eta =& \frac{f' (r-2M)}{\sqrt{r^2 - h' (r-2M)^2}}, \quad \eta \in (-\frac{\pi}{2}, \frac{\pi}{2}) \notag\\
=& \left(  \frac{r}{r-2M}\right)^{\halb} (\hat{H}r + \frac{c_1}{r^2})
\end{align}
Therefore, we can solve for $h$ as follows
\begin{align} \label{f-equation}
h' = \frac{(\tan \eta) r}{r-2M} \sqrt{\frac{1}{1+ \tan^2 \eta}}, \quad  \eta \in \quad \left(-\frac{\pi}{2}, \frac{\pi}{2} \right)
\end{align}
where $\tan \eta$ is now a known function of $r$ from \eqref{tan-eta-soln}.  We have the following structure of the second fundamental form, 
\begin{align}
\hat{K}_{11} = \hat{K}_{22}= \hat{H} + \frac{c_1}{r^2}, \quad \hat{K}_{33} = \hat{H} - \frac{2 c_1}{r^3}
\end{align}
where $c_1$ is a constant of integration in \eqref{f-equation}. Therefore, in the special case of $c_1 =0$ we have 
\begin{align}
\hat{K}_{11} = \hat{K}_{22} = \hat{K}_{33} =H
\end{align}
which corresponds to a spacetime umbilical slice. It may be noted that in our analysis, we are interested in the constant mean curvature slicing, which is convenient in solving \eqref{f-equation}. However, one can vary the constant mean curvaure $\hat{H}$ parameter so that one gets the desired properties. In particular, the following observation is relevant for our analysis: 
Consider the limit
\begin{align}
\lim_{r \to \infty}   \frac{(\tan \eta) r}{r-2M} \sqrt{\frac{1}{1+ \tan^2 \eta}}
\intertext{after plugging in \eqref{tan-eta-soln}}
\lim_{r \to \infty} \frac{Hr^{\frac{5}{2}}}{ r-2M \sqrt{r-2M + H^2 r^2} }
\end{align}
which becomes $1$ if $H>0$ or $-1$ if $H<0.$ If we consider the norm of the normal vector to the slice 
\begin{align}
\bar{g}^{\mu \nu} \ptl_\mu \hat{t} \ptl_\nu \hat{t} =& - \frac{r}{r-2M} + \frac{r-2M}{r} h'^2 \\
=& \frac{-r}{r-2M + \hat{H}^2 r^2}
\end{align} 
Therefore, as long as $\hat{H} \neq 0$, we have 
\begin{align}
\lim_{r \to \infty} \bar{g}^{\mu \nu} \ptl_\mu \hat{t} \ptl_\nu \hat{t} =0
\end{align}
As a consequence, the normal $n$ of the slices $\olin{\Sigma}_{\hat{t}}$ is asymptotically null-like as $r \to \infty$. However, we can relax this restriction i.e., we can construct asymptotically null-like slices that are not necessarily umbilical, using a different choice of the constant $c_1(\neq 0)$.  Let us now summarize the facts that we shall use for constant mean curvature hypersurfaces, based on the works of Lee-Lee \cite{Lee-Lee_11, Lee-Lee_16} (see also Brill-Cavallo-Isenberg\cite{BCIsen_80} and Malec-Murchadha \cite{MM_03, MM_09}). 

\begin{theorem}
Suppose $(\bar{M}, \bar{g})$ is a spherically symmetric maximal (Kruskal) extension of the Schwarzschild spacetime. Then there exist  constant mean curvature hypersurfaces, with the height function $h,$ such that 
\begin{enumerate}
\item The spacelike condition $ \bar{g}^{\mu \nu} \ptl_\mu \hat{t} \ptl_\nu \hat{t} <0$ for the CMC hypersurfaces is preserved at the event horizon
\item The CMC hypersurfaces become null-like asymptotically, i.e., as $r \to \infty$ where $r$ is the area radius function, for $\hat{H} \neq 0.$
\end{enumerate}
\end{theorem}
\section{Quasi-local Mass}
Let us return to the metric $\bar{g}$ in the $(\hat{t}, r, \theta, \phi)$ coordinates: 
\begin{align}
dt =& - d\hat{t} + f' dr 
\end{align}
the metric $\bar{g}$
\begin{align}
\bar{g} = - \hat{N} d \hat{t}^2 + \hat{q}_{ab} (dx^a + \hat{N}^a d \hat{t}) \otimes (dx^b + \hat{N}^b d \hat{t})
\end{align}
Now then consider the 2-surfaces $\Sigma \hookrightarrow \olin{\Sigma}_{\hat{t}}$ and let us compute the mean curvature vector $H$ in a coordinate form. Later we shall compute this in our orthonormal frame. We have
\begin{align}
 \mathbf{H} = H \vec{\nu}+ \text{tr}_{\Sigma} K \vec{n} 
\end{align}
where $\vec{\nu}$ and $\vec{n}$ are the (spatial) outward pointing and (timelike) future pointing normals of $\Sigma$ respectively.
The mean curvature vector $\mathbf{H}$ can be represented in the metric form as (compare with Section 3.2 in \cite{Wang_15})
\begin{align}
\mathbf{H} = \frac{2}{r} e^{-2Z} (\ptl_R r \ptl_R - \ptl_T r \ptl_T).
\end{align}
In this form it is also explicit that the null expansions are represented as the null derivatives of the (area radius) function $r$. 
As a consequence, we can now calculate the Hawking mass
\begin{align} \label{cmc-haw}
m_{H} =& \frac{\vert  \Sigma\vert^{1/2}}{(16 \pi)^{3/2}} \left( 16 \pi - \int_{\Sigma} H^2 \right), \notag\\
 =& \frac{ \vert \Sigma \vert^{1/2}}{(16 \pi)^{1/2}} ( 1 - (f^{-1} - f h'^2)^{-1}) \notag\\
 =& \frac{r}{2} (1-(f^{-1} - f h'^2)^{-1})
\end{align}
and likewise for the Brown-York mass, first note that there exists a unique (up to rigid motions) isometric embedding of the 2-surface $\Sigma$ (Gauss curvature $K>0$) into the Euclidean space $\mathbb{R}^3,$ from the classic embedding theorems (see e.g., \cite{Nirenberg_53,Pogorelov_52}). 

As a consequence, the mean curvature of the image of the isometric embedding of $\Sigma$ into the Euclidean space is $H_0 = \frac{2}{r}.$ Therefore, we have 
\begin{align}
\frac{1}{8 \pi} \int_{\Sigma} H_0 = r.
\end{align}  
As a consequence, 
\begin{align} \label{cmc-by}
m_{\text{BY}} =& \frac{1}{8 \pi} \int_{\Sigma} (H_0 -H) \bar{\mu}_{\Sigma} \notag\\
=& r- \frac{1}{8 \pi} \int_{\Sigma} \frac{2}{r}(f^{-1}- f h'^2)^{-1/2} \bar{\mu}_{\Sigma} \notag\\
=& r (1- (f^{-1} - f h'^2)^{-1/2})
\end{align}
and the Liu-Yau mass. Firstly note that,
\begin{align}
\vert \mathbf{H} \vert = \sqrt{H^2 - \text{tr}_{\Sigma} \hat{K}^2}
\end{align}
Therefore, 
\begin{align}\label{cmc-ly}
m_{\text{LY}} =& \frac{1}{8 \pi} \int (H_0 - \vert \mathbf{H} \vert ) \bar{\mu}_{\Sigma} \notag\\
=& r - \frac{1}{8 \pi}\int_{\Sigma} \left(\frac{4}{r^2} (f^{-1} - f h'^2)^{-1} - 4 \frac{(r-2M)^2h'^2}{r^4 p} \right)^{1/2} \bar{\mu}_{\Sigma} \notag\\
=& r -r \left(1 - f h'^2 \right)^{1/2} (f^{-1} - f h'^2)^{-1/2}.
\end{align}
We would like to remark that in these hypersurfaces, the distinction between the Brown-York mass \eqref{cmc-by} and the Liu-Yau \eqref{cmc-ly} mass becomes transparent. 

Let us now make few remarks on the positivity and compare the properties of the quasi-local masss introduced above. The Hawking mass is not necesarily positive.  The positivity of the Hawking mass was established for constant mean curvature 2-surfaces by Christodoulou-Yau. The positivity of Brown-York mass for 2-surfaces with $H >0$ (the mean curvature) and $K>0$ (Gauss curvature) was established by Shi-Tam \cite{shi2002}. Likewise, the positivity of the Liu-Yau mass  for $K>0$ and $H> \text{tr}_{\Sigma} \olin{K}$ was established in \cite{liu-yau-1,liu-yau-2}. The Liu-Yau mass is not frame dependent in a space-time sense. This aspect is becomes especially transparant in the non-maximal hypersurfaces considered in our work. The Liu-Yau mass is 'more positive' than the Brown-York mass, in fact we have the estimate $m_{\text{LY}} \geq m_{\text{BY}}.$  The Wang-Yau mass is based on the optimal isometric embeddings of the surface $\Sigma$ into the Minkowski space. The Wang-Yau mass of a sphere of symmetry $\Sigma$ in the exterior region of the Schwarzschild spacetime coincides with the Liu-Yau mass, because one can first show that $\tau=0$ is a solution of the optimal isometric embedding equation and the standard embedding into $\mathbb{R}^3$ (which corresponds to
 $\tau=0$ ) is local quasilocal energy minimizing (see \cite{Po-Wa-Yau_14} for notation and details).  It would be interesting to compute and analyse the Wang-Yau mass in our setting in the interior of Schwarzschild black hole spacetime. We propose to understand this in a future work.

Let us now compute the Ricci and Scalar curvatures of the hypersurface $\olin{\Sigma}.$  
Note that, in view of the Gauss-Kodazzi relations we have: 
\begin{subequations}
\begin{align}
\bar{R}\left({\frac{\ptl}{\ptl r}} , \frac{\ptl}{\ptl r} \right) =& \frac{2}{r} \frac{\ptl_r (f^{-1} - h'^2 f)}{ f^{-1} - h'^2f} \\
\bar{R}\left({\frac{\ptl}{\ptl \theta}} , \frac{\ptl}{\ptl \theta} \right)=& R \left({\frac{\ptl}{\ptl \theta}} , \frac{\ptl}{\ptl \theta} \right) + \frac{r}{(f^{-1} - h'^2 f)^2} \ptl_r  (f^{-1} - h'^2 f) + \frac{2}{f^{-1} - h'^2 f}  \\
\bar{R}\left({\frac{\ptl}{\ptl \phi}} , \frac{\ptl}{\ptl \phi} \right)=& R \left({\frac{\ptl}{\ptl \phi}} , \frac{\ptl}{\ptl \phi} \right) + \frac{r \sin^2 \theta}{(f^{-1} - h'^2 f)^2} \ptl_r  (f^{-1} - h'^2 f) + \frac{2 \sin^2 \theta}{f^{-1} - h'^2 f}
\end{align} 
for our spherically symmetric (warped product) spacetime.  It may be noted for the sphere of symmetry $(\Sigma, q)$ we have the Ricci curvature $R_{ab}= K q_{ab}.$
\end{subequations}
The remaining (off-diagonal) components of the Ricci tensor $\bar{R}_{ij}$ vanish for our form of the metric.  Likewise, we have the scalar curvature $\bar{R}$ of $\olin{\Sigma}$ as follows: 

\begin{align}
\bar{R} = \frac{2}{r^2}\left(1-\frac{1}{(f^{-1} - h'^2 f)^2} \right) +    \frac{4}{r (f^{-1} - h'^2 f)^{5/2}} \ptl_r (f^{-1} - h'^2 f) . 
\end{align}

\begin{lemma}
Consider the constant mean curvature hypersurfaces $\olin{\Sigma}$ such that the mean curvature $\hat{H} \in \mathbb{R},$ then the Hawking mass $m_{\textbf{H}}$ (eq. \eqref{cmc-haw}), Brown-York mass $m_{\text{BY}}$ (eq. \eqref{cmc-by}) and the Liu-Yau mass $m_{\text{LY}}$ (eq. \eqref{cmc-ly}) tend to the ADM mass $M$ as $\Sigma(r) \to \iota_0$ (spatial infinity); and the (spacetime) Hawking mass \eqref{spacetime-Hawking} and Liu-Yau mass tend to $M$ as $\Sigma(r) \to \mathcal{I}$ (null infinity). 
\end{lemma}
\begin{proof}
Let us start with the case $\hat{H} =0.$ In this case, the hypersurfaces $\olin{\Sigma}$ intersect with spatial infinity in the outer asymptotic region (as $r \to \infty$). We have, 

\begin{align}
m_{\text{H}} =& \frac{\vert \Sigma \vert^{1/2}}{(16 \pi)^{3/2}} \left( 16 \pi - \int \vert H \vert^2\right)  \\
=& \frac{r}{2} (1- (f^{-1} - f h'^2)^{-1}) \notag\\
=& \frac{r}{2} ( 1-f) =M 
\intertext{Brown-York mass}
m_{\text{BY}}=& r ( 1- (f^{-1} -f h'^2)^{-1/2}) \notag\\
=& r (1- f^{1/2}) =M
\intertext{Liu-Yau mass}
m_{\text{LY}} =& r ( 1- (f^{-1} -f h'^2)^{-1/2} (1-f^2 h'^2)^{1/2}) \notag\\
=& r (1-f^{1/2}) =M
\end{align}
as $\Sigma(r) \to \mathcal{I}, \hat{H} = 0.$ We would like to remind the reader that the $\hat{H} =0$ condition corresponds to $\olin{\Sigma}$ being asymptotically spacelike at the outer boundary. 
In case $\hat{H} \neq 0,$ recall that $\displaystyle \lim_{r \to \infty} h' = -1$ for $\hat{H}<0$ and $\displaystyle \lim_{r \to \infty} h' = 1$ for $\hat{H}>0.$ We have, 

\begin{align}
m_{\text{H}} =&  \frac{\vert \Sigma \vert^{1/2}}{(16 \pi)^{3/2}} \left( 16 \pi - \int \vert \mathbf{H} \vert^2\right)    \notag \\
=&  \frac{r}{2} \left( (1- (1-f^2 h'^2)(f^{-1} - f h'^2)^{-1}) \right) \notag\\
\intertext{now then as $\Sigma(r) \to \mathcal{I}$, we have}
 \lim_{ \Sigma(r) \to \mathcal{I}} m_{\text{H}}=& \frac{r}{2} \left( 1-\frac{f}{1-f^2} \cdot (1-f^2) \right) = M.
\intertext{Liu-Yau mass}
m_{\text{LY}} =& r ( 1- (f^{-1} -f h'^2)^{-1/2} (1-f^2 h'^2)^{1/2}) \notag\\
\intertext{as $\Sigma(r) \to \mathcal{I}$ we  have }
=& r \left(1- (1-f^2)^{1/2} \left(\frac{f}{1-f^2} \right)^{1/2} \right) = r (1-f^{1/2}) =M.
\end{align}
\end{proof}
In general, for a sphere of symmetry $\Sigma$ in a spherically symmetric spacetime, the spacetime version of the Hawking mass and the Liu-Yau mass are related as follows (see Section 3.2 in \cite{Wang_15})
\begin{align}
m_{\text{H}} (\Sigma(r)) = m_{\text{LY}} (\Sigma(r)) - \frac{m^2_{\text{LY}} (\Sigma(r))}{2r} .
\end{align}
In particular, 	it follows that as long as the area radius $r$ goes to infinity, the spacetime Hawking mass and the Liu-Yau mass have the same limit. The Wang-Yau quasi-local mass at null-infinity for the Vaidya spacetime was established in \cite{Po-Wa-Yau_16}.  We shall now focus our attention near the event horizon. Let us use the following notation, $\Sigma(r) \to \mathcal{H}^+$ denotes the limit as $\Sigma(r)$ approaches the event-horizon $\mathcal{H}$ from the domain of outer communications i.e., in the region $r \in (2M, 2M+\delta), \delta >0.$ Likewise, $\Sigma (r) \to \mathcal{H}^-$ denotes the limit as $\Sigma (r)$ approaches the event horizon $\mathcal{H}$ from the black hole interior, i.e., from the region $r \in (2M-\delta, 2M), \delta>0.$ 
\begin{lemma}
Suppose we have a regular solution to the constant mean curvature equation for  $c_1 \in \mathbb{R}$ in the domain of outer communications of the maximally extended Schwarzschild black hole spacetimes, then the Hawking mass $m_{\textbf{H}} $ (eq. \eqref{cmc-haw}), Brown-York mass $m_{\text{BY}} $ (eq. \eqref{cmc-by}) and the Liu-Yau mass $m_{\text{LY}} $ (eq. \eqref{cmc-ly}) are well-defined and have the limits, as $\Sigma(r) \to \mathcal{H}^{+},$ as follows. 
\end{lemma}
\begin{proof}
In the case $c_1 = -8M^3 \hat{H},$ it follows that the height function $h' = \mathcal{O} ((r-2M)^{-1/2})$ in the region $r \in (2M, 2M+ \delta)$ (see \cite{Lee-Lee_16}). In particular, we have
\begin{align}
h' = \hat{H} \left(\frac{r}{r-2M} \right)^{\halb} \left( \frac{r ( r^2 + 2Mr + 4M^2)^2}{ r^3 + \hat{H}^2 (r-2M) (r^2+2Mr + 4M^2)^2}\right)^{\halb}
\end{align}
thus, 

\begin{align}
m_{\text{H}} =& \frac{r}{2} (1- (f^{-1} -f h'^2)^{-1}) \notag\\
=& \frac{r}{2} \left( 1- \frac{f}{1- \hat{H}^2(r-2M) \left( \frac{r ( r^2 + 2Mr + 4M^2)^2}{ r^3 + \hat{H}^2 (r-2M) (r^2+2Mr + 4M^2)^2} \right) } \right) \notag\\
=& M \quad \text{as} \quad \Sigma(r) \to \mathcal{H}^{+}, \quad c_1 = - 8M^3\hat{H}.
\end{align}

 Likewise, if $c_1 < -8M^3 \hat{H}$ then $h' = \mathcal{O} ((r-2M)^{-1})$ in the region $r \in (2M, 2M+\delta), \delta >0$ is small. In particular, $h'$ is of the form
 \begin{align}
 h '= - \frac{1}{f} + \halb \frac{1}{ \left(f + (\hat{H} r + \frac{c_1}{r^2})^2 \right)} + \frac{1}{8}
 \frac{h}{ \left( f + (\hat{H} r + \frac{c_1}{r^2})^2 \right)^2} + \cdots
 \end{align}
 in the region $r \in (2M, 2M+\delta), \delta >0.$
   As a consequence,

\begin{align} \label{hawking-c1-upper}
m_{\text{H}} =&  \frac{r}{2} (1- (f^{-1} -fh'^2)^{-1}) \notag\\
\intertext{plugging in $h'$ in the region $r \in (2M, 2M+\delta), \delta >0$}
=& \frac{r}{2} \left(1- f \left(1 + \frac{1}{f} \left(\hat{H} r + \frac{c_1}{r^2} \right)^2 \right) \right)  \notag\\
=&M \left(1- \frac{c^2_1}{16M^4} - \frac{\hat{H}c_1}{M} - 4M^2\hat{H}^2 \right), \quad \text{as} \quad  \Sigma(r) \to \mathcal{H}^+, c_1 <-8M^3 \hat{H}. 
\end{align} 

On the other hand, for $c >-8M^3 \hat{H}$ in the region $r \in (2M, 2M+\delta), \delta >0,$ we have $h'$ of the form
\begin{align} \label{h-expansion}
h'= \frac{1}{f} - \halb \frac{1}{\left(f + (\hat{H} r + \frac{c_1}{r^2})^2 \right)} - \frac{1}{8} \frac{f}{ \left(f + (\hat{H} r + \frac{c_1}{r^2})^2  \right)^2} + \cdots
\end{align}
therefore,  we have
\begin{align}
m_{\text{H}}  =& \frac{r}{2} \left(1- f \left(1 + \frac{1}{f} \left(\hat{H} r + \frac{c_1}{r^2} \right)^2 \right) \right)  \notag\\
=&M \left(1- \frac{c^2_1}{16M^4} - \frac{\hat{H}c_1}{M} - 4M^2\hat{H}^2 \right), \quad \text{as} \quad  \Sigma(r) \to \mathcal{H}^+, c_1 > -8M^3 \hat{H}
\end{align}
analogous to \eqref{hawking-c1-upper}. 
Now then, for the Brown-York mass, for $c_1 = -8M^3 \hat{H},$
\begin{align}
m_{\text{BY}} =& r (1-(f^{-1} - fh'^2)^{-1/2}) \notag\\
=& r \left( 1- \left( \frac{f}{1-  \hat{H}^2(r-2M) \left( \frac{r ( r^2 + 2Mr + 4M^2)^2}{ r^3 + \hat{H}^2 (r-2M) (r^2+2Mr + 4M^2)^2} \right) } \right)^{\halb} \right) \notag\\
=&  2M \quad \text{as} \quad \Sigma(r) \to \mathcal{H}^{+}, c_1 = -8M^3 \hat{H}. 
\end{align}
Consider the case of constant mean curvature hypersurfaces with $c_1 <-8M^3 \hat{H},$
\begin{align}
m_{\text{BY}} =& r (1-(f^{-1} - fh'^2)^{-1/2})  \notag\\
\intertext{plugging in the behaviour of $h'$ for $c_1 < -8M^3 \hat{H}$}
 =& r \left(1- f^{\halb} \left(1+ \frac{1}{f} (\hat{H} r + \frac{c_1}{r^2})^2 \right)^{\halb} \right) \notag\\
=& 2M \left(1 - 2M \hat{H} - \frac{c_1}{4M^2} \right) , \quad \text{as} \quad \Sigma(r)\to \mathcal{H}^+, c_1<-8M^3\hat{H}.
\end{align}
On the other hand, for $c_1 >-8M^3 \hat{H},$ recall that $h' >0$ with the behaviour given by \eqref{h-expansion}. Therefore, we have
\begin{align}
m_{\text{BY}} =&  r (1-(f^{-1} - fh'^2)^{-1/2}) \notag\\
=&r \left(1- f^{\halb} \left(1+ \frac{1}{f} (\hat{H} r + \frac{c_1}{r^2})^2 \right)^{\halb} \right) \notag\\
=& 2M \left(1 - 2M \hat{H} - \frac{c_1}{4M^2} \right) , \quad \text{as} \quad \Sigma(r)\to \mathcal{H}^+, c_1>-8M^3\hat{H}.
\end{align}
Note that  $\displaystyle \lim_{\Sigma(r) \to \mathcal{H}^+} f =0$ with $f > 0$ in the region $ r \in (2M, 2M+ \delta)$.
Now let us turn to the Liu-Yau mass. We have for $c_1 = -8M^3 \hat{H},$
\begin{align}
m_{\text{LY}} =& r \big(1-(f^{-1} - f h'^2)^{-1/2}) (1-f^2 h'^2)^{1/2} \big) \notag\\
=& r \bigg( 1- \left( \frac{f}{1-  \hat{H}^2(r-2M) \bigg( \frac{r ( r^2 + 2Mr + 4M^2)^2}{ r^3 + \hat{H}^2 (r-2M) (r^2+2Mr + 4M^2)^2} \bigg) } \right)^{\halb} \cdot \notag\\
& \quad \quad \cdot \bigg(1- \hat{H}^2 \frac{(r-2M) ( r^2 + 2Mr + 4M^2)^2}{ r^3 + \hat{H}^2 (r-2M) (r^2 + 2Mr + 4M^2)} \bigg)  \bigg) \notag\\
\intertext{from which it is evident that}
m_{\text{LY}}= 2M 
\end{align}
as $\Sigma(r) \to \mathcal{H}^{+}$ for $c_1 = -8M^3 \hat{H}.$ Next, let us turn to the case, $c_1 <-8M^3 \hat{H}.$ We have
\begin{align}
m_{\text{LY}} =& r \left(1- (f^{-1} - f h'^2)^{-1/2} (1-f^2h'^2)^{1/2} \right) \notag\\
=& r \left(1- f^{1/2} \left(1+ \frac{1}{f} (\hat{H} r + \frac{c_1}{r^2})^2 \right)^{1/2} \left(\frac{1}{1+ \frac{1}{f} (\hat{H} r + \frac{c_1}{r^2})^2} \right)^{1/2} \right) \notag\\
=& r \left( \left(1+ 2M \hat{H} + \frac{c_1}{ 8M^2} \right) \left(\frac{1}{1+ \frac{1}{f} (\hat{H} r + \frac{c_1}{r^2})^2} \right)^{1/2} \right) \notag\\
=& 2M, \quad  \text{as} \quad \Sigma(r) \to \mathcal{H}^{+}, \quad c_1 < -8M^3\hat{H}.
\end{align}
Analogously, for the case $c_1 >-8M^3$ in the region $r \in (2M, 2M+ \delta), \delta >0$ we have
\begin{align}
m_{\text{LY}}=& r \left(1- f^{1/2} \left(1+ \frac{1}{f} (\hat{H} r + \frac{c_1}{r^2})^2 \right)^{1/2} \left(\frac{1}{1+ \frac{1}{f} (\hat{H} r + \frac{c_1}{r^2})^2} \right)^{1/2} \right) \notag\\
=& r \left( \left(1- 2M \hat{H} - \frac{c_1}{ 8M^2} \right) \left(\frac{1}{1+ \frac{1}{f} (\hat{H} r + \frac{c_1}{r^2})^2} \right)^{1/2} \right) \notag\\
\intertext{which also gives the limit}
=& 2M, \quad  \text{as} \quad \Sigma(r) \to \mathcal{H}^{+}, \quad c_1 > -8M^3\hat{H}.
\end{align}
\end{proof}
 In our analysis, we have made a choice of the constant mean curvature and the integration constant in the equation for $h$ \eqref{f-equation}, so that it has a desirable behaviour at infinity. However, we can, in principle, make another choice of these parameters and discuss the constant mean curvature hypersurfaces in the interior of the black hole. We would like to remind the reader that, in the usual interior Schwarzschild coordinates, the $t=const.$ hypersurfaces are no longer spacelike. Indeed, one can note that the $r=const.$ hypersurfaces are spacelike ($\ptl_r$ is a timelike vector) in the Schwarzschild interior. However, we shall follow the work of Lee-Lee\cite{Lee-Lee_16}, and introduce a height function to construct constant mean curvature hypersurfaces in the interior. In the Schwarzschild interior, we consider two cases of the height function $\mathfrak{h}.$
 
 \subsection{Case I: $\mathfrak{h}' >0$}

 Consider a time function $ \vec{t}$ such that $\vec{t} = \mathfrak{h} - t$ and we shall impose the condition $\bar{g}^{\mu \nu} \ptl_\nu \vec{t} \ptl_\mu \vec{t}= - f^{-1} + f \mathfrak{h'}^2 <0,$ for a timelike normal of $\vec{\Sigma}.$ Correspondingly, the orthonormal frame is given by
 \begin{align}
 \mbo{e}_3 = (f^{-1} - \mathfrak{h}^2 f )^{-1} (\mathfrak{h}' \ptl_t + \ptl_r), \quad \mbo{e}_4=  (f^{-1} - \mathfrak{h}^2 f )^{-1}(f^{-1} \ptl_t + \mathfrak{h}' f \ptl_r)
 \end{align}

The CMC equation, in the black hole interior is given by 

\begin{align} \label{CMC-int-h'-neg}
\mathfrak{h}'' + \left( ( f^{-1} - h'^2 f) ( \frac{2f}{r} + \frac{f'}{2}) + \frac{h'}{h} \right) \mathfrak{h}'
- 3 \hat{H} (f^{-1} - \mathfrak{h}'^2 f)^{3/2} =0, \quad \mathfrak{h}'>0. 
\end{align} 

\noindent Following \cite{Lee-Lee_11, Lee-Lee_16}, to the this equation , introduce the variable $\eta(r)$ such that  $\sec \eta = \mathfrak{h}' f$ where $\eta \in ( \pi/2, \pi ).$ The equation \eqref{CMC-int-h'-neg} can be explicitly solved as $\csc  \eta = -\frac{1}{(-f)^{1/2}} \left(- \hat{H} r - \frac{c_2}{r^2} \right) $ then we get 

\begin{align}
\mathfrak{h}' = \frac{1}{-f} \left(  \frac{\csc^2 \eta }{ \csc^2 \eta -1}\right), \eta \in(\pi/2, \pi), \quad \mathfrak{h}' >0. 
\end{align}

\subsection{Case II: $\mathfrak{h}' <0$}
Now Consider a time function $ \vec{t}$ such that $\vec{t} = -\mathfrak{h} + t$ and we shall impose the condition $\bar{g}^{\mu \nu} \ptl_\nu \vec{t} \ptl_\mu \vec{t} <0,$ for a timelike normal of $\vec{\Sigma}.$ Correspondingly, the orthonormal frame is given by
 \begin{align}
 \mbo{e}_3 = -(f^{-1} - \mathfrak{h}^2 f )^{-1} (\mathfrak{h}' \ptl_t + \ptl_r), \quad \mbo{e}_4=  -(f^{-1} - \mathfrak{h}^2 f )^{-1}(f^{-1} \ptl_t + \mathfrak{h}' f \ptl_r)
 \end{align}

The CMC equation, analogous to the previous expression in the black hole (interior) region, is given by 

\begin{align}
\mathfrak{h}'' + \left( ( f^{-1} - \mathfrak{h}'^2 f) ( \frac{2f}{r} + \frac{f'}{2}) + \frac{f'}{f} \right) \mathfrak{h}'
+ 3 \hat{H} (f^{-1} - \mathfrak{h}'^2 f)^{3/2} =0.
\end{align}  

In the interior this equation is solved by assigning 
\begin{align}
f \mathfrak{h}' =  \sec (\eta(r)), \quad \eta \in (0, \pi/2)
\end{align}
where 
\begin{align}
\csc \eta = \frac{1}{(-f)^{1/2}} \left(- Hr -\frac{c_2}{r^2} \right), \eta \in (0, \pi/2) \quad \text{so that} \quad
\mathfrak{h}' = \frac{1}{f} \sqrt{\frac{\csc^2 \eta}{\csc^2 \eta-1}}, \mathfrak{h}' <0,
\end{align}
where $c_2$ is a constant of integration. In contrast with the exterior, we have the following restriction on the existence of CMC in the interior: 
\begin{align}
\csc \eta \geq 1, \quad \text{so that} \quad k_H \geq c_2, \quad
k_H \fdg = - Hr^3 - r^{3/2} (2M-r)^{1/2}.
\end{align}
 
In this region,  if we choose $c_2 < -8M^3H$ then $\mathfrak{h}'$ is $\mathcal{O}((r-2M)^{-1} )$  and the spacelike CMC condition is preserved as $r \to 2M^{-}.$ This case allows us to consider smooth extensions of the constant mean curvature hypersurfaces into the Schwarzschild interior.

Let us now find the quasi-local mass quantities in the Schwarzschild interior. We would like to emphasize that in the standard coordinates in the Schwarzschild interior, the $t=const$ hypersurfaces are not necessarily spacelike. In the interior, we can construct `cylindrical' CMC hypersurfaces from $r=const$
hypersurfaces. 

However, we shall follow the construction where the introduction of the height function $\mathfrak{h}$ makes $\vec{t} =const.$ hypersurfaces manifestly spacelike. As a consequence, the previous formulas for quasi-local mass carry forward with appropriate modifications. Let us start with the Hawking mass: 
\begin{align}
m_{\text{H}} =&  \frac{\vert \Sigma \vert^{1/2}}{(16 \pi)^{3/2}} (16 \pi - \int_{\Sigma} H^2)  \notag\\
=& \frac{r}{2} (1- (f^{-1} - f \mathfrak{h}^2))^{-1}, \quad r<2M
\end{align}
and the Brown-York mass
\begin{align}
m_{\text{BY}} =& \int_{\Sigma} (H_0 - H ) \bar{\mu}_{\Sigma} \notag\\
=& r (1-(f^{-1} - f \mathfrak{h}'^2)^{-1/2}), \quad r<2M
\end{align}

Let us establish the blow up dates of the quasi-local masses $m_{\text{H}}$ and $m_{\text{BY}}$ at the singularity. Suppose $c_2$ is such that $\mathfrak{h}$ is well-defined in the Schwarzschild interior then. We have 

\begin{align}
\lim_{\Sigma(r) \to i^+  }m_{\text{H}} =& \lim_{\Sigma (r) \to i^+}  \frac{r}{2} (1- (f^{-1} - f \mathfrak{h}^2))^{-1} = \lim_{r \to 0} rf \left(1+ f^{-1} (\hat{H} f + \frac{c_2}{r^2})^2 \right) \\
=& \lim_{r \to 0} \left(rf + r (\hat{H} r + \frac{c_2}{r^2})^2 \right)
\end{align} 
Therefore, it follows that near the singularity $r \in (0, \delta), \delta >0$ the Hawking mass $m_{\text{H}}$ blows up like $\mathcal{O}(r^{-3})$ for $c_2 \neq 0.$ Likewise, let us analyse the behaviour of the Brown-York mass near the singularity. We have 

\begin{align}
\lim_{\Sigma(r) \to i^+  }m_{\text{BY}} =& \lim_{\Sigma (r) \to i^+} r \left(1-(f^{-1} - f \mathfrak{h}'^2)^{-1/2}\right) \notag\\
=& \lim_{r \to 0}  r \left( 1- \left(f - \left(\hat{H} r + \frac{c_2}{r^2} \right)^2 \right)^{1/2} \right),
\end{align}
from which it follows that the Brown-York mass $m_{\text{BY}}$ has $\mathcal{O}(r^{-1})$ blow up behaviour for $c_2 \neq 0$ near the singularity.  

\begin{lemma}
Suppose there exists solution of the CMC equation for $c_2 <-8M^3 \hat{H}$, then the quasi-local mass quantities $m_{\text{H}}, m_{\text{BY}}$  are continuous as $\Sigma(r) \to \mathcal{H}$  if the solutions of the constant mean curvature equations are such that $c_2 =c_1 (<-8M^3 \hat{H}).$ This condition for the continuity is consistent with a necessary condition for the smooth gluing of the constant mean curvature hypersurfaces in the interior and exterior of Schwarzschild spacetime \cite{Lee-Lee_16}
\end{lemma}

\begin{proof}
In the case where $c_2 < -8M^3 \hat{H},$ it may be noted that $\mathfrak{h}$ is well-defined in a left neightbourhood of the horizon $\mathcal{H}$ i.e., $r \in (2M-\delta, 2M).$  Furthermore, $\mathfrak{h}' = \mathcal{O} ((2M-r)^{-1}),$ $r \in (2M-\delta, \delta), \delta>0.$ In particular, for the case $\mathfrak{h}' >0,$ we have
\begin{align}
\mathfrak{h'} = -\frac{1}{f} + \frac{1}{2f^2 (Hr + \frac{c^2_2}{r^2})^2} + \cdots
\end{align}
in the region $r \in (2M-\delta, 2M), \delta>0.$ Now let us consider the Hawking mass as $\Sigma(r) \to \mathcal{H}^{-}.$
\begin{align}
m_{\text{H}} =& \frac{r}{2} (1-(f^{-1}-f\mathfrak{h}'^2)^{-1} ) \notag\\
=& \frac{r}{2} \left( 1- f \left(1-f^2 \left(\frac{1}{(-f)^2} \frac{\csc^2 \eta}{\csc^2 \eta -1} \right) \right)^{-1} \right), \quad \eta \in (\pi/2, \pi) \notag\\
=& \frac{r}{2}\left(1- f \left(1- \left(\frac{1}{(-f)^{1/2}} (-\hat{H} r - \frac{c_2}{r^2}) \right)^2 \right) \right) \notag\\
=&  \frac{r}{2} \left(  1+ f \left(-\frac{1}{f} (\hat{H}r + \frac{c_2}{r^2})^2-1 \right)\right), \notag\\
=& M \left(1- \frac{c^2_2}{16M^2} -  \frac{c_2 \hat{H}} {M} - 4M^2 \hat{H}^2 \right)
\end{align}
as $\Sigma(r) \to \mathcal{H}^-, c_2 <-8M^3 \hat{H}, \mathfrak{h}' >0.$ The behaviour of the Hawking mass $m_{\text{H}}$ is quite similar for $\mathfrak{h}' <0:$ 
\begin{align}
m_{\text{H}} =& \frac{r}{2} \left( 1- f \left(1-f^2 \left(\frac{1}{(-f)^2} \frac{\csc^2 \eta}{\csc^2 \eta -1} \right) \right)^{-1} \right), \quad \eta \in (0, \pi/2) \notag\\
=&  \frac{r}{2}\left(1- f \left(1- \left(\frac{1}{(-f)^{1/2}} (-\hat{H} r - \frac{c_2}{r^2}) \right)^2 \right) \right) \notag\\
=&  \frac{r}{2} \left(  1+ f \left(-\frac{1}{f} (\hat{H}r + \frac{c_2}{r^2})^2-1 \right)\right), \notag\\
=& M \left(1- \frac{c^2_2}{16M^2} -  \frac{c_2 \hat{H}}{M} - 4 M^2 \hat{H}^2 \right), \quad \Sigma(r) \to \mathcal{H}^-, c_2 <-8M^3 \hat{H}, \mathfrak{h}' <0.
\end{align}
Analogously, consider the Brown-York mass $m_{\text{BY}}$ as $\Sigma(r) \to \mathcal{H}^{-}.$ We have,

\begin{align}
m_{\text{BY}} =& r (1- (f^{-1} - f \mathfrak{h}'^2)^{-1/2}) \notag\\
=& r\left (   1- (-f)^{1/2} \left(  f^2  \left(\frac{\csc^2 \eta} {(-f)^2 (\csc^2 \eta -1)} \right) -1\right)^{-1/2} \right), \quad \eta \in (\pi/2, \pi) \notag\\
=& r \left(1-  \left(f \left(1+ \frac{1}{f} \left(\hat{H} r + c_2 \right)^2  \right) \right)^{1/2} \right) \notag\\
=& 2M \left(1- 2M \hat{H} - \frac{c_2}{4M^2} \right), \quad \text{as} \quad \Sigma (r) \to \mathcal{H}^{-}, c_2 <-8M^3 \hat{H}, \mathfrak{h}'>0.
\end{align}
Analogously, 
\begin{align}
m_{\text{BY}} =& r (1- (f^{-1} - f \mathfrak{h}'^2)^{-1/2}) \notag\\
=&  r\left (   1- (-f)^{1/2} \left(  f^2  \left(\frac{\csc^2 \eta} {f^2 (\csc^2 \eta -1)} \right) -1\right)^{-1/2} \right), \quad \eta \in (0, \pi/2) \notag\\
=& r \left(1- f \left(1+ \frac{1}{f} \left(\hat{H} r + c_2 \right)^2 \right)^{1/2} \right) \notag\\
=& 2M \left(1+ 2M \hat{H} + \frac{c_2}{4M^2} \right), \quad \text{as} \quad \Sigma (r) \to \mathcal{H}^{-}, c_2 <-8M^3 \hat{H}, \mathfrak{h}'<0.
\end{align}
The result on the continuity follows. 
\end{proof}
 
\subsection*{Acknowledgements}
N. Gudapati acknowledges the support of the Gordon and Betty Moore Foundation and the John Templeton Foundation through the Black Hole Initiative of Harvard Univeristy. S.-T. Yau acknowledges the support from NSF Grant DMS-1607871.  We thank Mu-Tao Wang for his comments on the article. 

\bibliography{central-bib}
\bibliographystyle{plain}
\end{document}